\documentclass{llncs}

\usepackage{enumerate}
\usepackage{amsmath,amsfonts,amssymb}

\newtheorem{algorithm}{Algorithm}

\newcommand{\set}[1]{\{ #1 \}}

\newcommand{\eset}[1]{\{ \: #1 \: \}}

\newcommand{\Ii}{{\cal I}}
\newcommand{\Jj}{{\cal J}}

\newcommand{\prob}[1]{\mathbb{P}\{#1\}}
\newcommand{\Prob}[1]{\mathbb{P}\big\{#1\big\}}
\newcommand{\Expe}[1]{\mathbb{E}\big\{#1\big\}}
\newcommand{\expe}[1]{\mathbb{E}\{#1\}}

\newcommand{\supp}{\mathrm{supp}}

\title{Approximate well-supported Nash equilibria \\ 
  in symmetric bimatrix games%
  \thanks{Partially supported by the Centre for Discrete Mathematics
    and its Applications (DIMAP) and EPSRC grant EP/D063191/1.}}  

\author{Artur Czumaj \and Michail Fasoulakis \and Marcin Jurdzi\'nski}
\institute{Centre for Discrete Mathematics and its Applications (DIMAP)\\
  Department of Computer Science, University of Warwick, UK \\
  \{A.Czumaj, M.Fasoulakis, M.Jurdzinski\}@warwick.ac.uk}

\begin{document}

\maketitle
\thispagestyle{plain}
\begin{abstract}
The $\varepsilon$-well-supported Nash equilibrium is a strong notion 
of approximation of a Nash equilibrium, where no player has an
incentive greater than $\varepsilon$ to deviate from any of the pure
strategies that she uses in her mixed strategy.
The smallest constant $\varepsilon$ currently known for which there is
a polynomial-time algorithm that computes an
$\varepsilon$-well-supported Nash equilibrium in bimatrix games is
slightly below~$2/3$.  
In this paper we study this problem for 
\emph{symmetric} bimatrix games and we provide a polynomial-time
algorithm that gives a $(1/2+\delta)$-well-supported Nash equilibrium, 
for an arbitrarily small positive constant~$\delta$.   
\end{abstract}

\section{Introduction}

The problem of computing Nash equilibria is one of the most
fundamental problems in algorithmic game theory. 
It is now known that the complexity of computing a Nash equilibrium is
PPAD-complete~\cite{DGP09}, even for two-player games~\cite{CDT09}.  
Given this evidence of intractability of the problem, further research
has focused on the computation of \emph{approximate} Nash equilibria.  
In this context---and assuming that all payoffs are normalized to be
in the interval $[0,1]$---the standard notion of approximation is the 
additive approximation with a parameter $\varepsilon \in [0,1]$. 
There are two different notions of additive approximation of
Nash equilibria: the \emph{$\varepsilon$-Nash equilibrium} and the
\emph{$\varepsilon$-well-supported Nash equilibrium}. 

An $\varepsilon$-Nash equilibrium is a strategy profile---one strategy
for each player---in which no player can improve her payoff by more
than $\varepsilon$ through unilateral deviation from her strategy in
the strategy profile. 
Several polynomial-time algorithms have been proposed to find
$\varepsilon$-Nash equilibria for $\varepsilon=1/2$~\cite{DMP09}, for
$\varepsilon = (3-\sqrt{5})/2 \approx 0.38$~\cite{DMP07}, for
$\varepsilon = 1/2 - 1/(3\sqrt{6}) \approx 0.36$~\cite{BBM10}, and
finally for $\varepsilon \approx 0.3393$~\cite{TS08}.  
It is also known how to find $\varepsilon$-Nash equlibria in
quasi-polynomial time $n^{O(\log n/\varepsilon^2)}$ for arbitrarily
small $\varepsilon>0$~\cite{LMM03}, where $n$ is the number of pure  
strategies.  

The notion of an $\varepsilon$-well-supported Nash equilibrium
requires that no player has an incentive greater than~$\varepsilon$ to
deviate from any of the pure strategies she uses in her mixed
strategy. 
It is a notion stronger than that of an $\varepsilon$-Nash
equilibrium: 
every $\varepsilon$-well-supported Nash equilibrium is also an
$\varepsilon$-Nash equilibrium, but not necessarily vice-versa. 
The smallest $\varepsilon$ for which a polynomial-time algorithm is
currently known that computes an $\varepsilon$-well-supported Nash 
equilibrium in an arbitrary bimatrix game is slightly
above~$0.6619$~\cite{KS10,FGSS12}.  
It is also known that for the class of win-lose bimatrix games one
can find $1/2$-well-supported Nash equilibria in polynomial
time~\cite{KS10}.  

In this paper we study computation of approximate well-supported Nash 
equilibria in \emph{symmetric} bimatrix games, a class of bimatrix
games in which swapping the roles of the two players does not change
the payoff matrices, that is if the payoff matrix of one is the
transpose of the payoff matrix of the other. 
Symmetric games are an important class of games in game theory;
their applications include auctions and congestion games. 
They have already been studied by Nash in his seminal paper in which
he introduced the concept of a Nash equilibrium;
he proved that every symmetric game has at least one symmetric Nash
equilibrium, that is one in which all players use the same mixed 
strategy~\cite{N51}.  

Computing Nash equilibria in symmetric bimatrix games is known to
be as hard as computing Nash equilibria in arbitrary bimatrix games 
because there is a polynomial-time reduction from the latter to the 
former~\cite{GKT50}.  
In contrast to arbitrary bimatrix games, it is known how to compute 
$(1/3+\delta)$-Nash equilibria in symmetric bimatrix games  
in polynomial time, where $\delta > 0$ is arbitrarily
small~\cite{KS11}.  
In this paper we improve our understanding of the approximability of
Nash equilibria in symmetric bimatrix games by considering the task of
computing approximate well-supported Nash equilibria. 
Our main result is an algorithm that computes
$(1/2+\delta)$-well-supported Nash equilibria in symmetric bimatrix
games in polynomial time, where $\delta>0$ is arbitrarily small 
(Theorem~\ref{thm:main}). 

Our $(1/2+\delta)$-approximation algorithm splits the analysis into
two cases that are then considered independently. 
The first case is based on the following relaxation of the concept of
a symmetric Nash equilibrium: we say that a strategy profile $(x, x)$
\emph{prevents exceeding} $u \in [0,1]$ if the expected
payoff of every pure strategy in the symmetric game is at most $u$
when the other player uses strategy~$x$. 
This is indeed a relaxation of the concept of the symmetric Nash
equilibrium because if $(x^*,x^*)$ is a symmetric Nash equilibrium
then it prevents exceeding its value (that is, the expected payoff
each player gets when they both play strategy~$x^*$).   
We justify relevance of this concept by showing that a strategy
profile $(x, x)$ that prevents exceeding~$u$ is a $u$-well-supported
Nash equilibrium, so in order to provide a latter it is sufficient to
find a former.  
Moreover, we show that this relaxation of a symmetric Nash equilibrium
is algorithmically tractable because it suffices to solve a single
linear program to find a strategy profile $(x, x)$ that prevents
exceeding~$u$, if there is one. 
The first case in our algorithm is to solve this linear program for 
$u = 1/2$ and if it succeeds then we can immediately report a
$1/2$-well-supported Nash equilibrium.   
Note that by the above, if there is indeed a symmetric Nash
equilibrium with value~$1/2$ or smaller, then the linear
program does have a solution.  

If the first case in the algorithm fails to identify a
$1/2$-well-supported equilibrium because the game has no
symmetric Nash equilibrium with value $1/2$ or smaller, then
we consider the other, and technically more challenging, case. 
We use another relaxation of the concept of a symmetric Nash
equilibrium: we say that a strategy profile $(x, y)$
\emph{well supports} $u \in [0, 1]$ if the expected payoff of 
every pure strategy in the support of~$x$ is at least~$u$ when the
other player uses strategy~$y$, and the expected payoff of every pure
strategy in the support of~$y$ is at least~$u$ when the other player
uses strategy~$x$. 
We observe that if a strategy profile $(x, y)$ well supports~$u$ then
it is a~$(1-u)$-well-supported Nash equilibrium, so in order to
provide a latter it is sufficient to find a former.

Therefore, in order to obtain a
$(1/2+\delta)$-well-supported Nash equilibrium, we are interested in 
finding a strategy profile $(x, y)$ that well supports 
$u \geq 1/2-\delta$. 
While it may not be easy to verify if there is such a strategy
profile, let alone find one, both can be achieved in polynomial time
by solving a single linear program if we happen to know the supports
of strategies of each player in such a strategy profile.
The obvious technical obstacle to algorithmic tractability here is that
the number of all possible supports to consider is exponential in the
number of pure strategies.
We overcome this difficulty by proving the main technical result of
the paper (Theorem~\ref{thm:semi-strongly-achieve}) that for every
symmetric Nash equilibrium $(x^*, x^*)$ and for every $\delta>0$
establishes existence of a strategy profile $(x, y)$, with both
strategies having supports of constant size, that well supports
$u^*-\delta$, where $u^*$ is the value of the Nash equilibrium. 
Note that by the failure of the first case every symmetric Nash
equilibrium has value larger than $1/2$, and hence
Theorem~\ref{thm:semi-strongly-achieve} implies that there is such a
strategy profile with constant-size supports that well supports
$1/2-\delta$.  
The second case of our algorithm is to solve the linear programs
mentioned above for $u=1/2-\delta$ and for all supports~$I$ and~$J$ of
sizes at most $\kappa(\delta)$---where $\kappa(\delta)$ is a constant  
(which depends on~$\delta$, but does not depend on the number~$n$
of pure strategies)
that is specified in Theorem~\ref{thm:semi-strongly-achieve}---and to 
output a solution $(x, y)$ as soon as one is found.

In order to prove our main technical result
(Theorem~\ref{thm:semi-strongly-achieve}) we use the probabilistic
method to prove existence of constant-support strategy profiles that
nearly well support the expected payoffs of a Nash equilibrium.   
Our construction and proof are inspired by the construction of
Daskalakis et al.~\cite{DMP07} used by them to compute
$(3-\sqrt{5})/2$-Nash equilibria in bimatrix games in polynomial time,
but our analysis is different and more involved because we need to
guarantee the extra condition of nearly well supporting the
equilibrium values. 
The general idea of using sampling and Hoeffding bounds to prove
existence of approximate equlibria with small supports dates back to
the papers of Althofer~\cite{Alt94} and Lipton et al.~\cite{LMM03},
who have shown that strategies with supports of 
size~$O(\log n/\varepsilon^2)$ are sufficient for $\varepsilon$-Nash
equilibria in games with~$n$ strategies.

\section{Preliminaries}

We consider bimatrix games $(R, C)$, where 
$R, C \in [0, 1]^{n \times n}$ are square matrices of payoffs for the
two players: the row player and the column player, respectively.
If the row player uses a strategy $i$, $1 \leq i \leq n$ and 
if the column one uses a strategy $j$, $1 \leq j \leq n$, then 
the row player receives payoff~$R_{ij}$ and the column player receives 
payoff~$C_{ij}$. 
We assume that the payoff values are in the interval $[0, 1]$;
it is easy to see that equilibria in bimatrix games are invariant
under additive and positive multiplicative transformations of the
payoff matrices.  

A \emph{mixed strategy}~$x \in [0, 1]^n$ is a probability distribution
on the set of \emph{pure strategies} $\eset{1, 2, \dots, n}$. 
If the row player uses a mixed strategy~$x$ and the column player
uses a mixed strategy~$y$, then the row player receives payoff 
$x^T R y$ and the column player receives payoff $x^T C y$. 
A pair of strategies $(x, y)$, the former for the row player and the
latter for the column player, is often referred to as a strategy
profile. 
We define the \emph{support} $\supp(x)$ of a mixed strategy~$x$ to be
the set of pure strategies that have positive probability in~$x$,
i.e., 
$\supp(x) = \set{i \: : \: 1 \leq i \leq n \text{ and } x_i > 0}$.   

For every $i$, $1 \leq i \leq n$, let $R_{i\bullet}$ be the row vector
of the payoffs of the payoff matrix $R$ when the row player uses
the strategy~$i$.  
Note that if the row player uses a pure strategy~$i$,  
$1 \leq i \leq n$, and if the column player uses a mixed
strategy~$y$, then the row player receives payoff~$R_{i \bullet} y$. 
Similarly, for every $j$, $1 \leq j \leq n$, let $C_{\bullet j}$ be 
the column vector of the payoffs of the matrix $C$ when the column
player uses the strategy~$j$.  
Note that if the column player uses a pure strategy~$j$, 
$1 \leq j \leq n$, and if the row player uses a mixed strategy~$x$,
then the column player receives payoff~$x^T C_{\bullet j}$. 

\begin{definition}[Nash equilibrium]
  A $\emph{Nash equilibrium}$ is a strategy profile $(x^*,y^*)$ such
  that
  \begin{itemize}
  \item
    for every $i$, $1 \leq i \leq n$, we have 
    $R_{i \bullet} y^* \leq (x^*)^T R y^*$, and
  \item
    for every $j$, $1 \leq j \leq n$, we have
    $(x^*)^T C_{\bullet j} \leq (x^*)^T C y^*$,
  \end{itemize}
  or, in other words, 
  if $x^*$ is a best response to~$y^*$ and $y^*$ is a best
  response to~$x^*$. 
\end{definition}

\begin{definition}[Approximate Nash equilibrium]
  For every $\varepsilon > 0$, an 
  \emph{$\varepsilon$-Nash equilibrium} is a strategy profile
  $(x^*,y^*)$ such that
  \begin{itemize}
  \item
    for every $i$, $1 \leq i \leq n$, we have 
    $R_{i \bullet} y^* - (x^*)^T R y^* \leq \varepsilon$, and 
  \item
    for every $j$, $1 \leq j \leq n$, we have
    $(x^*)^T C_{\bullet j} - (x^*)^T C y^* \leq \varepsilon$, 
  \end{itemize}
  or, in other words, 
  if $x^*$ is an $\varepsilon$-best response to~$y^*$ and $y^*$ is an
  $\varepsilon$-best response to~$x^*$. 
\end{definition}

\begin{definition}[Approximate well-supported Nash equilibrium] 
  For every $\varepsilon > 0$, an 
  \emph{$\varepsilon$-well-supported Nash equilibrium} is a strategy
  profile $(x^*,y^*)$ such that
  \begin{itemize}
  \item
    for every $i$, $1 \leq i \leq n$, and $i' \in \supp(x^*)$, we
    have $R_{i \bullet} y^* - R_{i' \bullet} y^* \leq \varepsilon$, 
    and 
  \item
    for every $j$, $1 \leq j \leq n$, and $j' \in \supp(y^*)$, we have 
    $(x^*)^T C_{\bullet j} - (x^*)^T C_{\bullet j'} \leq \varepsilon$, 
  \end{itemize}
  or, in other words, 
  if every~$i' \in \supp(x^*)$ is an $\varepsilon$-best response
  to~$y^*$ and every~$j' \in \supp(y^*)$ is an $\varepsilon$-best
  response to~$x^*$. 
\end{definition}

\begin{definition}[Symmetric game, symmetric Nash equilibrium]
  A bimatrix game $(R, C)$ is \emph{symmetric} if $C = R^T$.  

  A \emph{symmetric Nash equilibrium} in a symmetric bimatrix game
  $(R, R^T)$ is a strategy profile $(x^*, x^*)$ such that
  for every $i$, $1 \leq i \leq n$, we have 
  $R_{i \bullet} x^* \leq (x^*)^T R x^*$.
  Note that then it also follows that for every $j$, 
  $1 \leq j \leq n$, we have:
  \[
  (x^*)^T R^T_{\bullet j} 
  = R_{j \bullet} x^*
  \leq (x^*)^T R x^*
  = (R x^*)^T x^*
  = (x^*)^T R^T x^*.
  \]
\end{definition}

Let us recall a fundamental theorem of Nash~\cite{N51} about existence
of symmetric Nash equilibria in symmetric bimatrix games.
\begin{theorem}[\cite{N51}]
  Every symmetric bimatrix game has a symmetric Nash equilibrium.  
\end{theorem}

\section{Computing approximate well-supported Nash equilibria} 

Fix a bimatrix game $G = (R, C)$ for the rest of the paper, where 
$R, C \in [0, 1]^{n \times n}$.  
We will use~$N$ to denote the number of bits needed to represent the
matrices~$R$ and~$C$ with all their entries represented in binary.  
We say that a strategy~$x$ is \emph{$k$-uniform}, 
for $k \in \mathbb{N} \setminus \eset{0}$, if  
$x_i \in \{0, \frac{1}{k},\frac{2}{k}, \dots, 1\}$, for every~$i$, 
$1 \leq i \leq n$. 

\subsection{Strategies that prevent exceeding a payoff} 

\begin{definition}[Preventing exceeding payoffs]
  We say that a strategy $x \in [0, 1]^n$ for the row player
  \emph{prevents exceeding} $u \in [0, 1]$ if for every  
  $j = 1, 2, \dots, n$, we have $x^T C_{\bullet j} \leq u$ or, in other
  words, if the column player payoff of the best response to~$x$
  does not exceed~$u$.
  Similarly, we say that a strategy $y \in [0, 1]^n$ for the column
  player \emph{prevents exceeding} $v \in [0, 1]$ if for every   
  $i = 1, 2, \dots, n$, we have $R_{i \bullet} y \leq v$ or, in other 
  words, if the row player payoff of the best response to~$y$ 
  does not exceed~$v$.

  For brevity, we say that a strategy profile $(x, y)$ 
  \emph{prevents exceeding} $(v, u)$ if $x$ prevents exceeding~$u$ and
  $y$ prevents exceeding~$v$. 
\end{definition}

Observe that the following system of linear constraints
$\mathrm{PE}(v, u)$ characterizes strategy profiles $(x, y)$ that
prevent exceeding $(v, u) \in [0, 1]^2$: 
\begin{eqnarray*}
  \sum_{i=1}^n x_i = 1; & & 
  x_i \geq 0 \text{ for all } i = 1, 2, \dots, n; \\
  \sum_{j=1}^n y_j = 1; & & 
  y_j \geq 0 \text{ for all } j = 1, 2, \dots, n; \\
  R_{i \bullet} y \leq v & & \text{for all } i = 1, 2, \dots, n; \\ 
  x^T C_{\bullet j} \leq u & & \text{for all } j = 1, 2, \dots, n. 
\end{eqnarray*}
Note that if $(x, y)$ is a Nash equilibrium then, by definition, it
prevents exceeding $(x^T R y, x^T C y)$, which implies the following
Proposition. 

\begin{proposition} 
  \label{prop:PE-solution}
  If $(x, y)$ is a Nash equilibrium, $v \geq x^T R y$, and 
  $u \geq x^T C y$, then $\mathrm{PE}(v, u)$ has a solution and it 
  prevents exceeding $(v, u)$. 
\end{proposition} 

By the following proposition, in order to find an
$\varepsilon$-well-supported Nash equilibrium it suffices to find a 
strategy profile that prevents exceeding 
$(\varepsilon, \varepsilon)$. 

\begin{proposition}
  \label{prop:wsNe-from-pe}
  If a strategy profile $(x, y)$ prevents exceeding $(v, u)$ then it
  is a $\max(v, u)$-well-supported Nash equilibrium. 
\end{proposition}

\begin{proof}
  Let $i' \in \supp(x)$ and let $i \in \eset{1, 2, \dots, n}$. 
  Then we have:
  \[
  R_{i \bullet} y - R_{i' \bullet} y 
  \leq R_{i \bullet} y
  \leq v,
  \]
  where the first inequality follows from $R_{i' \bullet} y \geq 0$,
  and the other one holds because $y$ prevents exceeding~$v$. 
  Similarly, and using the assumption that $x$ prevents exceeding~$u$,
  we can argue that for all $j' \in \supp(y)$ and 
  $j \in \eset{1, 2, \dots, n}$, we have 
  $x^T C_{\bullet j} - x^T C_{\bullet j'} \leq u$. 
  It follows that $(x, y)$ is a $\max(v, u)$-well-supported Nash
  equilibrium. 
  \qed
\end{proof}

\subsection{Strategies that well support a payoff}

\begin{definition}[Well supporting payoffs]
  We say that a strategy $x \in [0, 1]^n$ for the row player
  \emph{well supports} $v \in [0, 1]$ against a strategy 
  $y \in [0, 1]^n$ for the column player if for every 
  $i \in \supp(x)$, we have $R_{i \bullet} y \geq v$. 
  Similarly, we say that a strategy $y \in [0, 1]^n$ for the column 
  player \emph{well supports} $u \in [0, 1]$ against a strategy 
  $x \in [0, 1]^n$ for the row player if for every $j \in \supp(y)$,
  we have $x^T C_{\bullet j} \geq u$.  

  For brevity, we say that a strategy profile $(x, y)$ 
  \emph{well supports} $(v, u)$ if $x$ well supports~$v$ against~$y$ and
  $y$ well supports~$u$ against~$x$. 
\end{definition}

The following theorem states that the payoffs of every Nash
equilibrium can be nearly well supported by a strategy profile with
supports of constant size.  
\begin{theorem}
  \label{thm:semi-strongly-achieve}
  Let $(x^*, y^*)$ be a Nash equilibrium. 
  For every $\delta > 0$, there are $\kappa(\delta)$-uniform
  strategies $x, y$ such that the strategy profile $(x, y)$
  well supports 
  $\big((x^*)^T R y^* - \delta, (x^*)^T C y^* - \delta\big)$,  
  where $\kappa(\delta) = \lceil 2\ln(1/\delta)/\delta^2 \rceil$.
\end{theorem}
The proof of this technical result is postponed until
Section~\ref{section:proof-of-thm}.

Let $v, u \in [0, 1]$, $\delta > 0$, and let~$\Ii$ and~$\Jj$ be
multisets of pure strategies of size~$\kappa(\delta)$.
Consider the following system $\mathrm{WS}(v, u, \Ii, \Jj, \delta)$ of
linear constraints: 
\begin{eqnarray*}
  x_i = k_i/\kappa(\delta) & & \text{for all } i = 1, 2, \dots, n; \\
  y_j = \ell_j/\kappa(\delta) & & \text{for all } j = 1, 2, \dots, n; \\
  R_{i \bullet} y \geq v - \delta & & \text{for all } i \in \Ii; \\
  x^T C_{\bullet j} \geq u - \delta & &
    \text{for all } j \in \Jj;
\end{eqnarray*}
where $k_i$ is the number of times~$i$ occurs in
multiset~$\Ii$, and $\ell_j$ is the number of times~$j$ occurs in
multiset~$\Jj$.
Note that the system $\mathrm{WS}(v, u, \Ii, \Jj, \delta)$ of linear
constraints characterizes $\kappa(\delta)$-uniform strategy profiles 
$(x, y)$, such that $\supp(x) = \Ii$ and $\supp(y) = \Jj$, that well 
support $(v-\delta, u-\delta)$. 
Theorem~\ref{thm:semi-strongly-achieve} implies the following.

\begin{corollary}
  \label{cor:WS-solution}
  If $(x, y)$ is a Nash equilibrium, $v \leq x^T R y$, 
  $u \leq x^T C y$, and $\delta > 0$, then there are multisets $\Ii$
  and $\Jj$ from $\eset{1, 2, \dots, n}$ of size~$\kappa(\delta)$, such
  that $\mathrm{WS}(v, u, \Ii, \Jj, \delta)$ has a solution and it
  well supports $(v-\delta, u-\delta)$. 
\end{corollary}

By the following proposition, in order to find an
$\varepsilon$-well-supported Nash equilibrium it suffices to find a 
strategy profile that well supports $(1-\varepsilon, 1-\varepsilon)$. 

\begin{proposition}
  \label{prop:wsNe-from-ws}
  If a strategy profile $(x, y)$ well supports $(v, u)$ then it
  is a $\big(1-\min(v, u)\big)$-well-supported Nash equilibrium.  
\end{proposition}

\begin{proof}
  Let $i' \in \supp(x)$ and let $i \in \eset{1, 2, \dots, n}$. 
  Then we have:
  \[
  R_{i \bullet} y - R_{i' \bullet} y 
  \leq 1 - R_{i' \bullet} y
  \leq 1 - v,
  \]
  where the first inequality follows from $R_{i \bullet} y \leq 1$,
  and the other one holds because $y$ well supports~$v$. 
  Similarly, and using the assumption that $x$ well supports~$u$, 
  we can argue that for all $j' \in \supp(y)$ and 
  $j \in \eset{1, 2, \dots, n}$, we have 
  $x^T C_{\bullet j} - x^T C_{\bullet j'} \leq 1-u$. 
  It follows that $(x, y)$ is a 
  $\big(1-\min(v, u)\big)$-well-supported Nash equilibrium. 
  \qed
\end{proof}

\subsection{The algorithm for symmetric games}

Propositions~\ref{prop:wsNe-from-pe} and~\ref{prop:wsNe-from-ws}
suggest that in order to identify a $1/2$-well-supported Nash 
equilibrium it suffices to find either a strategy profile that
prevents exceeding $(1/2, 1/2)$ or one that well supports 
$(1/2, 1/2)$.  
Moreover, verifying existence and identifying such strategy profiles
can be done efficiently by solving the linear program
$\mathrm{PE}(1/2, 1/2)$, and by solving linear programs 
$\mathrm{WS}(1/2+\delta, 1/2+\delta, \Ii, \Jj, \delta)$ for all
multisets $\Ii$ and $\Jj$ of pure strategies of size~$\kappa(\delta)$,
respectively.  

For arbitrary bimatrix games the above scheme may fail if none of
these systems of linear constraints has a solution.
Note, however, that---by Proposition~\ref{prop:PE-solution} and
Corollary~\ref{cor:WS-solution}---it would indeed succeed if we could
guarantee that the game had a Nash equilibrium with both payoffs at
most~$1/2$, or with both payoffs at least~$(1/2+\delta)$. 
Symmetric bimatrix games nearly satisfy this requirement thanks to
existence of symmetric Nash equilibria in every symmetric
game~\cite{N51}.

If $(x^*, x^*)$ is a symmetric Nash equilibrium in a symmetric
bimatrix game $(R, R^T)$ then---trivially---either 
$(x^*)^T R x^* \leq 1/2$ or $(x^*)^T R x^* > 1/2$. 
In the former case, by Proposition~\ref{prop:PE-solution} the linear
program $\mathrm{PE}(1/2, 1/2)$ has a solution, and by
Proposition~\ref{prop:wsNe-from-pe} it is a $(1/2)$-well-supported
Nash equilibrium. 
In the latter case, by Corollary~\ref{cor:WS-solution} there are
multisets~$\Ii$ and~$\Jj$ of pure strategies of size~$\kappa(\delta)$, such
that $\mathrm{WS}(1/2, 1/2, \Ii, \Jj, \delta)$ has a solution $(x, y)$
and it well supports $(1/2-\delta, 1/2-\delta)$.  
It then follows by Proposition~\ref{prop:wsNe-from-ws} that $(x, y)$
is a $(1/2+\delta)$-well-supported Nash equilibrium.

\begin{algorithm}
  \label{alg:half-approx}
  Let $(R, R^T)$ be a symmetric game and let $\delta > 0$.
  \begin{enumerate}
  \item
    \label{step:sce}
    If $\mathrm{PE}(1/2, 1/2)$ has a solution~$x$ then return 
    $(x, x)$.
  \item
    \label{step:ssa}
    Otherwise, that is if $\mathit{PE}(1/2, 1/2)$ does not have a 
    solution:
    \begin{enumerate}
    \item
      \label{step:guess}
      Using exhaustive search, find
      multisets~$\Ii$ and $\Jj$ of pure strategies,
      both of size~$\kappa(\delta)$,
      such that $\mathit{WS}(1/2, 1/2, \Ii, \Jj, \delta)$ has 
      a solution.
    \item
      \label{step:return}
      Return a solution $(x, y)$ of
      $\mathit{WS}(1/2, 1/2, \Ii, \Jj, \delta)$. 
      \qed
    \end{enumerate}
  \end{enumerate}
\end{algorithm}

In order to find appropriate $\Ii$ and $\Jj$ in step~2(a), 
an exhaustive enumeration of all pairs of multisets~$\Ii$ and~$\Jj$
of size~$\kappa(\delta)$ is done, and for each such pair the system 
of linear constraints $\mathit{WS}(1/2, 1/2, \Ii, \Jj, \delta)$ is
solved. 
Note that the number of $\kappa(\delta)$-element multisets from an
$n$-element set is
\[
{n+\kappa(\delta)-1 \choose \kappa(\delta)} =
n^{O(\kappa(\delta))} = n^{O(\ln(1/\delta)/\delta^2)}.
\]
Therefore, step~2.\ of the algorithm requires
solving $n^{O(\ln(1/\delta)/\delta^2)}$ linear programs
and hence the algorithm runs in time $N^{O(\ln(1/\delta)/\delta^2)}$. 

\begin{theorem}
  \label{thm:main}
  For every $\delta > 0$, Algorithm~\ref{alg:half-approx} runs in time 
  $N^{O(\ln(1/\delta)/\delta^2)}$ and it returns a strategy profile that
  is a $(1/2+\delta)$-well-supported Nash equilibrium. 
\end{theorem}

\section{Proof of Theorem~\ref{thm:semi-strongly-achieve}}
\label{section:proof-of-thm}

We use the probabilistic method:
random $\kappa(\delta)$-uniform strategies are drawn by
sampling $\kappa(\delta)$ pure strategies (with replacement) from the 
distributions~$x^*$ and~$y^*$, respectively, and Hoeffding's
inequality is used to show that the probability of thus selecting a
strategy profile that  well supports 
$\big(v^* - \delta, u^* - \delta\big)$ is positive if
$\kappa(\delta) \geq 2 \ln(1/\delta)/\delta^2$, where 
$v^* = (x^*)^T R y^*$ and $u^* = (x^*)^T C y^*$. 

Consider $2\kappa(\delta)$ mutually independent random
variables~$I_t$ and~$J_t$, $1 \leq t \leq \kappa(\delta)$,
with values in $\eset{1, 2, \dots, n}$, the former with the same
distribution as strategy~$x^*$ and the latter with the same
distribution as strategy~$y^*$, that is we have
$\prob{I_t=i} = x^*_i$ and $\Prob{J_t=j} = y^*_j$
for $i, j = 1, 2, \dots, n$.
Define the random distributions $X = (X_1, X_2, \dots, X_n)$ and
$Y = (Y_1, Y_2, \dots, Y_n)$, with values in $[0, 1]^n$, by setting:
\[
X_i =
\frac{1}{\kappa(\delta)} \cdot \sum_{t=1}^{\kappa(\delta)} [I_t=i]
\qquad \text{ and } \qquad 
Y_j =
\frac{1}{\kappa(\delta)} \cdot \sum_{t=1}^{\kappa(\delta)} [J_t=j].
\]
Note that every realization of~$Y$ is a $\kappa(\delta)$-uniform
strategy that uses the pure strategy~$j$, $1 \leq j \leq n$,
with probability $K_j/\kappa(\delta)$, where
$K_j = \sum_{t=1}^{\kappa(\delta)} [J_t=j]$ is the number of
indices~$t$, $1 \leq t \leq \kappa(\delta)$, for which $J_t = j$.
A similar characterization holds for every realization of~$X$.
Observe also that $\supp(X) \subseteq \supp(x^*)$ and 
$\supp(Y) \subseteq \supp(y^*)$ because for all~$i$ and~$j$, 
$1 \leq i, j \leq n$, the random variables~$X_i$ and~$Y_j$ are
identically equal to~$0$ unless $x^*_i > 0$ and $y^*_j > 0$,
respectively. 

Since we want (a realization of) the random strategies~$X$ and~$Y$
to well support a certain pair of values, we now characterize 
$R_{i \bullet} Y$, for all $i \in \supp(x^*)$;
the whole reasoning presented below for $R_{i \bullet} Y$ can be
carried out analogously for~$X^T C_{\bullet j}$,
for all $j = 1, 2, \dots, n$, and hence it is omitted.

First, observe that for all~$i = 1, 2, \dots, n$, we have:
\[
R_{i \bullet} Y
\, = \, \sum_{j=1}^n R_{ij} Y_j
\, = \, \frac{1}{\kappa(\delta)} \cdot \sum_{j=1}^n R_{ij} \cdot
\sum_{t=1}^{\kappa(\delta)} [J_t=j]
\, = \, \frac{1}{\kappa(\delta)} \cdot \sum_{t=1}^{\kappa(\delta)}
R_{i J_t}.
\]
Therefore, the random variable $R_{i \bullet} Y$ is equal to the
arithmetic average
\[
\overline{Z_i} =
\frac{1}{\kappa(\delta)} \cdot \sum_{t=1}^{\kappa(\delta)} Z_{it}
\]
of the independent random variables
$Z_{it} = R_{i J_t}$, $1 \leq t \leq \kappa(\delta)$.

For every $i \in \supp(x^*)$, we will apply Hoeffding's inequality
to the corresponding random variable $\overline{Z_i}$.
Hoeffding's inequality gives an exponential upper bound for the
probability of large deviations of the arithmetic average of
independent and bounded random variables from their expectation.

\begin{lemma}[Hoeffding's inequality]
  Let $Z_1, Z_2, \dots, Z_k$ be independent random variables with
  $0 \leq Z_t \leq 1$ for every $t$, let
  $\overline{Z} = (1/k) \cdot \sum_{t=1}^k Z_t$, and let
  $\Expe{\overline{Z}}$ be its expectation.
  Then for all $\delta > 0$, we have
  $\Prob{\overline{Z}-\Expe{\overline{Z}} \leq -\delta}
  \leq e^{-2 \delta^2 k}$.
\end{lemma}

Before we apply Hoeffding's inequality to the random
variables~$\overline{Z_i}$ defined above, observe that for every
$t = 1, 2, \dots, \kappa(\delta)$, we have:
\[
\expe{Z_{it}}
\, = \, \expe{R_{i J_t}}
\, = \, \sum_{j=1}^n R_{ij} \cdot \prob{J_t=j}
\, = \, R_{i \bullet} y^*.
\]
Note, however, that if $i \in \supp(x^*)$ then
$\expe{Z_{it}} = R_{i \bullet} y^* = v^*$, because $(x^*, y^*)$ is a
Nash equilibrium, and hence every $i \in \supp(x^*)$ is a best
response to~$y^*$.  
It follows that
$\Expe{\overline{Z_i}} =
(1/\kappa(\delta)) \cdot \sum_{t=1}^{\kappa(\delta)} \expe{Z_{it}} =
v^*$.

Applying Hoeffding's inequality, for every $i \in \supp(x^*)$, we
get:
\begin{equation}
  \label{eq:hoeff-i}
  \prob{R_{i \bullet} Y < v^* - \delta}
  \, = \,
  \Prob{\overline{Z_i} - \Expe{\overline{Z_i}} < -\delta}
  \, \leq \,
  e^{-2 \delta^2 \kappa(\delta)}.
\end{equation}
It follows that if $I \subseteq \supp(x^*)$ and
$|I| \leq \kappa(\delta)$, then:
\begin{multline}
  \label{eq:hoeff-I}
  \Prob{R_{i \bullet} Y < v^* - \delta
    \text{ for some } i \in I}
  \, \leq \, \\
  \leq \,
  \sum_{i \in I}
  \Prob{R_{i \bullet} Y < v^* - \delta}
  \, \leq \, \kappa(\delta) \cdot e^{-2 \delta^2 \kappa(\delta)}
  \, = \, 2 \delta^2 \ln(1/\delta)
  \, < \, \frac{1}{2},
\end{multline}
for all $\delta > 0$.
The first inequality holds by the union bound, and the second
follows from~(\ref{eq:hoeff-i}) and because
$|I| \leq \kappa(\delta)$.
The last inequality can be verified by observing that the
function $f(x) = 2 x^2 \ln(1/x)$, for $x > 0$, achieves its maximum
at $x = 1/\sqrt{e}$ and $f(1/\sqrt{e}) = 1/e < 1/2$.

In a similar way we can prove that if $J \subseteq \supp(y^*)$
and $|J| \leq \kappa(\delta)$, then:
\begin{equation}
  \label{eq:hoeff-J}
  \Prob{X^T C_{\bullet j} < (x^*)^T C y^* - \delta
    \text{ for some } j \in J}
  \, < \, \frac{1}{2},
\end{equation}
for all $\delta > 0$.

We are now ready to argue that
\begin{multline*}
  \Prob{R_{i \bullet} Y \geq v^* - \delta
    \text{ for all } i \in \supp(X), \\
    \text{ and } 
    X^T C_{\bullet j} \geq u^* - \delta
    \text{ for all } j \in \supp(Y)} 
  > 0,
\end{multline*}
and hence there must be realizations $x, y \in [0, 1]^n$ of the
random variables
$X = (X_1, X_2, \dots, X_n)$ and $Y = (Y_1, Y_2, \dots, Y_n)$, such
that~$(x, y)$ well supports 
$\big(v^* - \delta, u^* - \delta\big)$. 
Indeed, we have:
\begin{multline*}
  \Prob{R_{i \bullet} Y < v^* - \delta
    \text{ for some } i \in \supp(X), \\
    \text{ or } 
    X^T C_{\bullet j} < u^* - \delta
    \text{ for some } j \in \supp(Y)} \\
  \leq \,
  \sum_{I \subseteq \supp(x^*)}
  \Prob{I = \supp(X) \text{ and }
    R_{i \bullet} Y < v^* - \delta
    \text{ for some } i \in I} \\
  + \sum_{J \subseteq \supp(y^*)}
  \Prob{J = \supp(Y) \text{ and }
    X^T C_{\bullet j} < u^* - \delta
    \text{ for some } j \in J} \\
  = \,
  \sum_{\substack{I \subseteq \supp(x^*) \\
      |I| \leq \kappa(\delta)}}
  \Prob{I = \supp(X)} \cdot 
  \Prob{R_{i \bullet} Y < v^* - \delta
    \text{ for some } i \in I \; \big| \; 
    I = \supp(X)} \\
  + \sum_{\substack{J \subseteq \supp(y^*) \\
      |J| \leq \kappa(\delta)}}
  \Prob{J = \supp(Y)} \cdot 
  \Prob{X^T C_{\bullet j} < u^* - \delta
    \text{ for some } j \in J \; \big| \; 
    J = \supp(Y)} \\
  < \,
  \sum_{I \subseteq \supp(x^*)}
  \Prob{I = \supp(X)} \cdot \frac{1}{2}
  + \sum_{J \subseteq \supp(y^*)}
  \Prob{J = \supp(Y)} \cdot \frac{1}{2}
  \, = \, 1,
\end{multline*}
where the first inequality follows from the union bound, and 
from~$\supp(X) \subseteq \supp(x^*)$  
and~$\supp(Y) \subseteq \supp(y^*)$; 
the equality holds because~$|\supp(X)| \leq \kappa(\delta)$ and  
$|\supp(Y)| \leq \kappa(\delta)$
by the definitions of~$X$ and~$Y$;  
and the latter (strict) inequality follows from~(\ref{eq:hoeff-I})
and~(\ref{eq:hoeff-J}).

\subsection*{Acknowledgements}

We thank the anonymous SAGT reviewers for detailed feedback that
helped us improve the presentation of our results.

\bibliographystyle{plain}
\bibliography{agt}

\end{document}